\documentclass[a4paper, 12pt]{article} 
\usepackage[centertags]{amsmath}
\usepackage{amsmath}
\usepackage[dutch,british]{babel}
\usepackage{amsfonts}
\usepackage{amssymb} 
\usepackage{amsthm}
\usepackage{newlfont}
\usepackage{color}
\usepackage{subfig}
 \usepackage{bbm}
 \usepackage{float}
 \usepackage{import}

\usepackage{stmaryrd}
\usepackage{mathrsfs}
\usepackage{pstricks,pst-node}
\usepackage{palatino}  
\usepackage{fancyhdr}
\usepackage{graphicx}
\usepackage{fancybox}
\usepackage{geometry}
 \usepackage{psfrag}
 
 \usepackage{multirow}

\theoremstyle{plain}
  \newtheorem{theorem}{Theorem}[section]
  
  \newtheorem{proposition}[theorem]{Proposition}
  \newtheorem{lemma}[theorem]{Lemma}
  \newtheorem{remark}[theorem]{Remark}
\theoremstyle{definition}
  \newtheorem{definition}{Definition}[section]
  \newtheorem{assumption}[theorem]{Assumption}

\theoremstyle{remark}

\numberwithin{equation}{section}

\DeclareMathOperator{\Tr}{Tr}

 \DeclareMathOperator{\supp}{Supp}

\newcommand\otimesal{\mathop{\hbox{\raise 1.6 ex
  \hbox{$\scriptscriptstyle\mathrm{al}$}
\kern -0.92 em \hbox{$\otimes$}}}}
\newcommand\oplusal{\mathop{\hbox{\raise 1.6 ex
  \hbox{$\scriptscriptstyle\mathrm{al}$}
\kern -0.92 em \hbox{$\oplus$}}}}
\newcommand\Gammal{\hbox{\raise 1.7 ex
\hbox{$\scriptscriptstyle\mathrm{al}$}\kern -0.50 em $\Gamma$}}
\renewcommand\i{\mathrm{i}}

\let\al=\alpha \let\be=\beta  \let\ep=\epsilon

  \let\ga=\gamma 
\let\ka=\kappa \let\la=\lambda \let\om=\omega

 \let\Ga=\Gamma  \let\Om=\Omega

\newcommand{\caB}{{\mathcal B}}
\newcommand{\caC}{{\mathcal C}}
\newcommand{\caD}{{\mathcal D}}

\newcommand{\caK}{{\mathcal K}}

\newcommand{\caO}{{\mathcal O}}

\newcommand{\caW}{{\mathcal W}}

\newcommand{\scrH}{{\mathscr H}}

\newcommand{\bbC}{{\mathbb C}}

\newcommand{\bbN}{{\mathbb N}}

\newcommand{\bbR}{{\mathbb R}}

\newcommand{\opunit}{\text{1}\kern-0.22em\text{l}}

\newcommand{\frg}{{\mathfrak g}}
\newcommand{\frh}{{\mathfrak h}}

\newcommand{\e}{{\mathrm e}}

\renewcommand{\d}{{\mathrm d}}
\newcommand{\sys}{{\mathrm S}}
\newcommand{\res}{{\mathrm F}}

\newcommand{\Ran}{\mathrm{Ran}}

\newcommand{\Dom}{\mathrm{Dom}}

\newcommand{\beq}{ \begin{equation} }
\newcommand{\eeq}{ \end{equation} }
\newcommand{\baq}{\begin{eqnarray}}
\newcommand{\eaq}{\end{eqnarray}}
\newcommand{\bet}{ \begin{theorem} }
\newcommand{\eet}{ \end{theorem} }

\newcommand{\lone}{\mathbbm{1}}
\renewcommand{\supp}{\mathrm{Supp}}

\newcommand{\norm}{ \|}
\newcommand{\str}{ |}
\newcommand{\inter}{\mathrm{I}}

 \newcounter{smallarabics}
\newenvironment{arabicenumerate}
{\begin{list}{{\normalfont\textrm{\arabic{smallarabics})}}}
  {\usecounter{smallarabics}\setlength{\itemindent}{0cm}
  \setlength{\leftmargin}{5ex}\setlength{\labelwidth}{4ex}
  \setlength{\topsep}{0.75\parsep}\setlength{\partopsep}{0ex}
   \setlength{\itemsep}{0ex}}}
{\end{list}}

\newcounter{smallroman}

\newcommand{\ben}{\begin{arabicenumerate}}
\newcommand{\een}{\end{arabicenumerate}}

\newcommand{\smallv}{v_1}
\newcommand{\largev}{v_2}
\newcommand{\extension}{\mathrm{ex}}
\newcommand{\gs}{\mathrm{gs}}

\begin{document}

\begin{center}
\Large{ \bf{Asymptotic completeness  for the massless spin-boson model}}
 \\
\vspace{15pt} \normalsize
{\bf   W.  De Roeck\footnote{
email: {\tt w.deroeck@thphys.uni-heidelberg.de}}  }\\
\vspace{10pt} 
{\it   Institut f\"ur Theoretische Physik  \\ Universit\"at Heidelberg \\
Philosophenweg 16,  \\
D69120 Heidelberg,  Germany 
} \\

\vspace{15pt} \normalsize
{\bf M. Griesemer\footnote{
email: {\tt marcel@mathematik.uni-stuttgart.de  }}  }\\
\vspace{10pt} 
{\it   Department of Mathematics \\
Universit{\"a}t Stuttgart \\ 
Pfaffenwaldring 57\\ D-70569 Stuttgart, Germany 
} \\

\vspace{15pt}

{\bf   A. Kupiainen\footnote{
email: {\tt antti.kupiainen@helsinki.fi}}  }\\
\vspace{10pt} 
{\it   Department of Mathematics \\
 University of Helsinki \\ 
P.O. Box 68, FIN-00014,  Finland 
} \\

\vspace{15pt}

\end{center}

\begin{abstract}
We consider generalized versions of the massless spin-boson model.  Building on the recent work in \cite{deroeckkupiainenphotonbound} and \cite{deroeckkupiainenpropagation}, we prove asymptotic completeness.  \end{abstract}

\section{Introduction}
This paper is concerned with the scattering theory for generalized spin-boson
models with massless bosons. That is, we consider a spin system (an ``atom'') coupled
to a scalar field of quantized massless bosons. With the help of
previously established properties, such as relaxation to the ground
state and a uniform bound on the number of soft bosons, we now show that
asymptotic completeness holds provided the excited states of the
uncoupled system of spin and bosons have finite life times once the interaction is
turned on. (Fermi-Golden Rule condition.)

To describe our main result and its proof we now introduce the system
in some detail. We confine ourselves to a concrete system satisfying
all our assumptions. More general hypotheses are described in the next section. 
Our model consists of a small system (atom, spin) coupled to a free bosonic field.  The Hilbert space of the total system is 
$$
    \scrH = \scrH_\sys \otimes \scrH_\res
$$
where  $\scrH_\sys = \bbC^n$ for some $n<\infty$ ($\sys$ for ``small
system'') and the field space $\scrH_\res$ is the symmetric Fock space over $L^2(\bbR^3)$.
The total Hamiltonian is of the form
$$
H= H_\sys \otimes \lone+ \lone \otimes H_\res + H_{\inter}
$$
where $H_\sys$ is a hermitian matrix with simple eigenvalues only, and
$H_\res $ denotes the Hamiltonian of the free \emph{massless}
field. The coupling  operator $H_{\inter}$ is of the form 
$$
H_{\inter}= \la D \otimes \int_{\bbR^3}\big(\hat\phi(k)a^*_k+\overline{ \hat\phi(k)}a_k\big)\d k.
$$
Here $D=D^*$ is a matrix acting on $\scrH_\sys$,  $\la \in \bbR$ is a
sufficiently small coupling constant,  and $a^*_{k}, a_{k} $ are the usual creation and annihilation
operators of a mode $k\in \bbR^3$ satisfying the ``Canonical
Commutation Relations''.  The function $\hat\phi \in L^2(\bbR^3)$ is a ``form factor'' that
imposes some infrared regularity and an ultraviolet cutoff: simple
and sufficient assumptions are that $\hat\phi$ has compact support and
that $\hat\phi\in C^3(\bbR^3\backslash\{0\})$ with 
$$
        \hat\phi(k) = |k|^{(\alpha-1)/2},\qquad |k|\leq 1,
$$ 
for some $\alpha>0$. This assumption ensures, e.g., that $H$ has a unique
ground state $\Psi_{\gs}$. To rule out the existence of excited bound states we assume the
Fermi-Golden Rule condition stated in the next section. Further
spectral input is not needed but our results do certainly have non-trivial
consequences for the spectrum of $H$.

Under the time-evolution generated by $H$ it is expected that every excited state relaxes to the ground state by
emission of photons whose dynamics is asymptotically free.
The existence of excited states with this property is well known \cite{FGSch2001,GZ2009}; they
are spanned by products of asymptotic creation operators applied to
the ground state $\Psi_{\gs}$, that is, by vectors of the form:
$$
      a_{+}^{*}(f_1)\ldots  a_{+}^{*}(f_m)\Psi_{\gs} = \lim_{t\to\infty} e^{i(H-E)t} a^{*}(f_{1,t})\ldots a^{*}(f_{m,t})\Psi_{\gs},
$$
where $f_{t}= e^{-i\omega t}f$ and $\omega(k)=|k|$.
Asymptotic completeness of Rayleigh scattering means that the span of
these vectors is dense in $\scrH$.
In particular, the representation of the CCR given by the asymptotic
creation and annihilation is equivalent to the Fock representation; non-Fock representation such as those discussed in \cite{DG2004} do not occur. 

Asymptotic completeness of Rayleigh scattering is expected to hold for a large class of 
models of atoms interacting with quantized radiation. Yet, despite considerable efforts \cite{HS1995,DG1999,FGSch2002,Gerard2002}, it has so far resisted a
rigorous proof. The main stumbling blocks were the lack of time-independent photon bounds and a the lack of a quantitative understanding of a property called \emph{relaxation to the ground state} (the zero temperature analog of return to equlibrium), which is known to follow from asymptotic completeness \cite{FGSch2001}. 
For the spin-boson model, these problems were solved in
the previous papers  \cite{deroeckkupiainenphotonbound, deroeckkupiainenpropagation}. We now use these results and expand on them
to give a complete proof of asymptotic completeness for a fairly large
class of \emph{massless} spin-bosons system.
Independently from us Faupin and Sigal had embarked on a similar project
using results from \cite{deroeckkupiainenphotonbound}. They obtain AC
for a class of spin-boson models very similar to ours.

Our proof of asymptotic completeness relies on methods and tools developed for the purpose of establishing
the relaxation to equilibrium  (and non-equilibrium steady states) of systems at positive temperature.  These methods are based on analogues with one-dimensional statistical mechanics. In fact, by a cluster expansion we have previously shown (in \cite{deroeckkupiainenphotonbound})
that a weak form of relaxation to the ground state holds, by which we mean that
\beq\label{key1}
  \lim_{t  \to \infty}  \langle \Psi_t, O\Psi_t \rangle =    \langle
\Psi_{\gs},  O \Psi_{\gs} \rangle
\eeq
for a suitable $C^{*}$ algebra of observables $O$. Our second key
tool, established by the same cluster expansion, states that the number
of soft photons emitted in the process of relaxation is uniformly
bound in time. More precisely, there exits a constant $\kappa>0$ such
that
\beq\label{key2}
\sup_{t }  \langle \Psi_t, \e^{\ka N} \Psi_t \rangle  <\infty
\eeq
where $N$ denotes the number operator on Fock space.
Our third key ingredient concerns the number of bosons in a ball $|x|<
vt$ where $v<1$, the speed of light. We show that 
\beq\label{key3}
    \langle\Psi_{t}, \d\Ga (\theta_t)\Psi_{t}\rangle =
    \langle\Psi_{\gs}, \d\Ga (\theta_t)\Psi_{\gs}\rangle +  \caO( \langle t \rangle ^{-\al})
\eeq
where $\theta_t$ denotes a smoothed characteristic function of the set
$\{x\in\bbR^3\mid |x|\leq vt\}$ with some $v<1$. Moreover, if
$\theta_t$ is replaced by $\theta_{t_c}$ then \eqref{key3} holds
uniformly in $t_c$ with $|\lambda|^{-2} \leq t_c\leq t$.
This propagation bound and soft-photon bound were established in
\cite{deroeckkupiainenpropagation} by a slight variation of the
cluster expansion mentioned above (in the case of the soft-photon
bound, this adjustment was even unnecessary and one could have simply
copied the treatment of \cite{deroeckkupiainenphotonbound}, as was
also remarked in \cite{faupinsigalcommentnumber}) 

Our basic strategy for proving asymptotic completeness is the usual one from time-dependent
scattering theory \cite{SigalSoffer1987,Graf1990}. In the present
context this means that we construct a suitable (right-)inverse $Z$
of the wave operator $W_{+}$ define by
$$
   W_{+} a^{*}(f_1)\ldots  a^{*}(f_m)\Omega = a_{+}^{*}(f_1)\ldots  a_{+}^{*}(f_m)\Psi_{gs}.
$$
The identity $W_{+}Z=\lone$ in $\scrH$ shows that $W_{+}$ has full range, which
is equivalent to asymptotic completeness as explained above. Our tools, however, are
certainly not the usual ones: the properties \eqref{key1}-\eqref{key3} augmented by
a strong form of the local relaxation
\eqref{key1}, see Proposition~\ref{prop: fate of atom}, are enough
for establishing both the existence of $Z$ and the identity
$W_{+}Z=\lone$. Important standard tools from scattering theory, such as
propagation estimates and Mourre estimates, are not needed in the present work.

We conclude this introduction with a discussion of previous work on Rayleigh scattering and asymptotic
completeness (AC) in models similar to ours.
In '97 Spohn considered an  electron that is bound by a perturbed
harmonic potential and coupled to the quantized radiation field in
dipole approximation \cite{Spohn1997}.
Knowing AC for the harmonically bound particle, a result due to Arai \cite{Arai1983},
he concludes AC in his perturbed
model by summing a Dyson series for the inverse wave operator. This is
the first proof of AC in a non-solvable model with massless bosons.
AC in more general models of a bound particle coupled to massive
bosons was established by \cite{DG1999} and by \cite{FGSch2002}. These works adapt the
methods from many-body quantum scattering, i.e. Mourre estimates and
propagation estimates, to non-relativistic QFT.
Consequently the Mourre estimate for massless bosons was a main
concern of subsequent work \cite{Skibsted1998, GGM, FGSig1}. For
interesting partial results on massless boson scattering assuming a
Mourre estimate we refer to \cite{Gerard2002}. 

Most recently, in \cite{faupinsigalhuygens, faupinsigalcommentnumber},  Faupin and Sigal succeeded in establishing AC for a
large class of models of atoms interacting with quantized radiation.
They assume a
time-independent  bound on   $ \langle\Psi_{t}, \d\Ga (1/\str k \str)^2\Psi_{t}\rangle$.  Based on this assumption they
establish AC for a large variety of models including non-relativistic
QED and spin-boson models.
Following the strategy in \cite{deroeckkupiainenphotonbound},  a bound on $ \langle\Psi_{t}, \d\Ga (1/\str k \str)^2\Psi_{t}\rangle$ can be obtained in the same way as a bound on  $\langle\Psi_{t}, N^2  \Psi_{t}\rangle$ upon slightly strenghtening the infrared assumption, as is remarked in \cite{faupinsigalcommentnumber}.
For a certain class of spin-boson models (similar to ours), Faupin and Sigal have
thus proven AC prior to us.
However, given the methods and results from
\cite{deroeckkupiainenphotonbound}, which are needed by
them as well, we believe our approach to AC for spin-boson
models is simpler.

\subsubsection*{Acknowledgements}
M.G. thanks I.M. Sigal for explanations on his work with Faupin.  W.D.R. acknowledges the support of the DFG and  A.K. is supported by the ERC and the Academy of Finland.

\section{Assumptions and Results}
 \label{sec: model and results}

We now describe the class of spin-boson models considered in this paper along with all
assumptions. It is instructive to do this in $d$ rather than 3 space
dimension, although we shall later confine ourselves to 3 dimensions for simplicity.
This section also serves us for collecting frequently used notations.

\subsection{Notations}

The symmetric Fock-space over a one-particle space $\frh$ is denoted
by $\Ga(\frh)$ which is defined by 
\beq
   \Ga(\frh)  =\mathop{\oplus}\limits_{n=0}^\infty  P_{S} \frh^{\otimes_{n}} 
\eeq
with $P_S$ is the projection to symmetric tensors and $
\frh^{\otimes_0}\equiv \bbC $. Hence $\scrH_\res = \Ga(\frh)$ with $\frh=L^2(\bbR^d)$.
We define the finite-photon space
\beq
\caD_{fin}:= \cup_{n \in \bbN} \lone_{N\leq n} \scrH_\res,
\eeq
and we write also $\caD_{fin}$ for $\scrH_\sys \otimes \caD_{fin}$.

By $a^{*}(f)$ and $a(f)$ we denote the smeared creation and
annihilation operators in $ \scrH_\res$ that are related to $a^*_{k}$
and $a_{k}$ by 
$$
      a^{*}(f) = \int f(k) a^{*}_{k}\, dk,\qquad
      a(f) = \int \overline{f(k)} a_{k}\, dk.
$$
Their sum is the self-adjoint field operator $\Phi(f) =  a^{*}(f) +
a(f)$, which generates the Weyl-operator
\beq 
\caW(f) = \e^{\i \Phi(f)}.
\eeq
Our assumptions on form factor and initial states are conveniently expressed
in terms of the following spaces:
\begin{definition} For $0<\al<1$ the subspace $\frh_{ \al} \subset \frh$
consists of all $\psi \in \frh$  such that $\hat\psi\in  C^3(\bbR^d\setminus\{0\})$
 has compact support, and, for all multi-indices $m$ with $|m|\leq 3$,
\beq
\label{betadecay}
|\partial^m_k\hat\psi(k)|\leq C|k|^{(\beta-d+2)/2-|m|}
\eeq
for  some $\be>\al$ and $C <\infty$. By an easy application of Lemma A.1 of \cite{deroeckkupiainenpropagation},  any $\psi \in \frh_\al$ satisfies
\beq
\str \psi(x)\str \leq C (1+\str x \str)^{-\min\{ \frac{d+\beta+2}{2},3\}}
\eeq
for some $\beta >\al$.

The subspace $\caD_\al \subset \scrH$ is defined by 
\beq 
\caD_\al :=  \mathrm{Span} \{ \psi_\sys \otimes \caW(f) \Om \mid \psi_s\in \scrH_S,\ f\in \frh_{\al} \}.
\eeq 
It is dense in $ \scrH$ because $\frh_{ \al}$ is dense in $\frh$.
\end{definition}

The following notation is very convenient and often used in this paper:
If $\Psi\in \scrH =  \scrH_a\otimes  \scrH_b$ and $\gamma\in \scrH_a$,
then $\Psi_\gamma=\langle \gamma\str \Psi \in  \scrH_b$ 
is defined by 
$$
\eta \otimes \langle \gamma\str \Psi = \big(|\eta\rangle\langle\gamma|\otimes \lone\big)\Psi, \quad \text{for}\ \eta\in \scrH_a.
$$
It follows that $\langle \eta\otimes\Psi_\gamma, \Phi\rangle = \langle\Psi_\gamma,\Phi_\eta\rangle$ and 
$\|\langle \gamma\str \Psi  \|^2 = \langle\Psi,(|\gamma\rangle\langle\gamma|\otimes \lone)\Psi\rangle$.

\subsection{Model and Assumptions}

Recall from the introduction that the system we consider is described
by a Hamiltonian $H$ on $\scrH = \scrH_\sys \otimes \scrH_\res$
that is of the form
\beq \label{def: tot ham}
    H= H_\sys \otimes \lone+ \lone \otimes H_\res + H_{\inter}.
\eeq
The free field operator $H_\res$ and the interaction operator $H_{\inter}$ are given by  
\begin{align*}
    H_\res &= \d\Gamma(\omega)= \int_{\bbR^d} \omega(k) a^*_{k} a_{k}\d k,\qquad \omega(k)=|k|,\\ 
    H_{\inter} &= \la D \otimes \Phi(\phi). 
\end{align*}

In the following we describe our assumptions on the matrices $H_\sys$, $D$ and the
form factor $\phi$. These choices and assumptions are assumed to hold throughout the article with the exception of the present Section \ref{sec: model and results}, where we will state some earlier results that do not require these strong assumptions.
The infrared (small Fourier mode $k$) behavior of the form factor determines temporal correlations in the
model and some regularity near $k=0$ is needed:

\begin{assumption}[$\al$-Infrared regularity] \label{ass: infrared behaviour}
The form factor $\phi$  is in $\frh_\al$ and  the dimension $d=3$.
\end{assumption}

Of course, the restriction to $d=3$ is not really necessary provided one modifies slightly the infrared assumption but we prefer to keep it for simplicity.
Our second assumption ensures that the coupling is effective:
 \begin{assumption}[Fermi Golden Rule] \label{ass: fermi golden rule}
We assume that  the spectrum of $H_\sys$ is non-degenerate (all eigenvalues are simple) and we let $e_0 := \min \sigma( H_\sys)$ (atomic ground state energy). Most importantly, we assume that for any  eigenvalue $e \in \sigma( H_\sys), e \neq e_0$, there is a set $\{e_i\}_{i=1}^n$ of eigenvalues
such that 
\beq
e= e_n >  e_{n-1} > \ldots >  e_1> e_0,  
\eeq
and for all $ i =0,\ldots, n-1$,
\beq
 \Tr[P_i DP_{i+1}DP_i] \int_{\bbR^d}  \d k  \, \delta(\str k \str- e_i-e_{i+1})\str \hat \phi(k)\str^2 > 0. \label{eqexpression jump rates}
\eeq
Here $P_{i}$ denotes the projection onto the eigenspace associated
with $e_i$. The right hand side is well-defined since $\hat\phi$ is continuous.
\end{assumption}

 
\subsection{Results}

The following results on the ground state and its properties are prerequisites for 
our definition of the wave operator.

\bet[\cite{BFS1998,gerardgroundstate}]
Let the form factor $\phi$ satisfy  $\phi/\omega\in \frh$. Then the  ground-state energy 
\beq
E_{\gs} = \inf_{\|\Psi\|=1} \langle \Psi,  H \Psi \rangle
\eeq
is finite; $E_{\gs} \leq \inf H_\sys$, and there is a  vector $\Psi_{\gs} \in \scrH$ such that 
\beq
 H \Psi_{\gs}=  E_{\gs} \Psi_{\gs} 
\eeq
If the coupling constant $\str \la \str$ is sufficiently small, then such a $\Psi_{\gs}$ is unique (up to scalar multiplication). 
\eet

The uniqueness of the ground state for small  $\str \la \str$ follows
from the simplicity of the eigenvalue  $e_0 := \min \sigma( H_\sys)$
by an overlap argument that is probably due to \cite{BFS1998}.
 
Additionally, we need an exponential localization property of photons, namely:

\begin{lemma}\label{lem: exp phonon loc}
Assume Assumptions \ref{ass: infrared behaviour} and \ref{ass: fermi golden rule} and let $\str\la\str$ be sufficiently small. Then, for some $\ka>0$,  
\beq
\langle \Psi_{\gs}, \e^{\ka N} \Psi_{\gs} \rangle < \infty.
\eeq
\end{lemma}

This property is proven in this paper, at the end of Section~\ref{sec: strong  local relaxation}.
Earlier results, exhibiting somewhat weaker
localization, can be found in \cite{Hiroshima2003}. 

To describe scattering, let us introduce the standard identification operator $I: \scrH_\res \to \scrH$
\beq \label{def: identification i}
I a^{\ast}(f_1) a^{\ast}(f_2)  \ldots  a^{\ast}(f_m)  \Om : =  a^{\ast}(f_1) a^{\ast}(f_2)  \ldots  a^{\ast}(f_m)  \Psi_{\gs}
\eeq
for $f_1, \ldots, f_m $ such that $(1+ \str k \str^{-1})f_j \in \frh$, and then extended to a closed operator, see e.g. \cite{HS1995}. Note that the adjoint of 
$I$ is densely defined and hence $I$ is closable. 
The well-known Cook's  argument yields the following standard result:  

\begin{theorem}\label{thm: w op}
The wave operators
\beq \label{eq: limit of w operator}
W_{\pm}\Psi := \lim_{t\to\infty} W_{t}\Psi, \qquad     W_{t}\Psi = \e^{\mp \i t H} I  \e^{\pm \i t (E_{\gs}+H_\res)} \Psi
\eeq
exist for $\Psi $ in a dense subset of $\caD_{fin}$. They satisfy $\norm W_{\pm}\Psi \norm= \norm \Psi \norm $  and therefore $W_{\pm}$ extends to an isometry $\scrH_\res \to \scrH$.
\end{theorem}

Existence of  wave operators for a large class of models with massive and massless photons is
established in \cite{HK1969,DG1999,FGSch2001, GZ2009}.

Our main result complements the former statement by establishing
\bet[Asymptotic completeness] \label{thm: ac}
Assume Assumptions \ref{ass: infrared behaviour} and \ref{ass: fermi golden rule}. Then there is a $\la_0>0$ such that for any coupling strength $\la$ satisfying $0 < \str \la \str \leq \la_0$;
\beq
\Ran W_{\pm}  = \scrH.
\eeq 
\eet

\begin{remark} We will only discuss  $W_+$. To treat $W_-$, it suffices to define a time-reversal operator (antilinear involution)
\beq
\Theta= \Theta_\sys \otimes \Theta_\res
\eeq
where $\Theta_\sys$ is complex conjugation in a basis diagonalizing $H_\sys$, and $\Theta_\res= \Gamma(\theta_{\res})$ with  $\theta_{\res}\psi(q)= \overline{\psi(-q)}$.    Then we check that, if $H$ satisfies our assumptions, then $\Theta H \Theta$ satisfies them as well. 
\end{remark}


\section{Technical tools} \label{sec: technical tools}

In this section, we list all the tools that we use which are not proven in the present paper, or only in the appendix. 
By smooth indicators $\theta$, we mean spherically symmetric functions
in $C^{\infty}_0(\bbR^d) $, i.e.\  $\theta(x)=\theta(\str x \str)$,
with $0 \leq \theta \leq 1$,  and we write $\theta_t(x)=\theta(x/t)$. All the following lemmata hold for sufficiently small $\la$, uniformly
in $\la$. From now on, we always assume that  Assumptions \ref{ass:
  infrared behaviour} and \ref{ass: fermi golden rule} hold. Moreover,
we write throughout this section 
$$
      \Psi_t= \e^{-\i t H}\Psi_0\quad\text{ with}\quad  \Psi_0 \in \caD_\al. 
$$
The first tool is the weak relaxation to the ground state.
\begin{lemma}[Weak relaxation \cite{deroeckkupiainenphotonbound}]\label{lem: weak relaxation}
Let the $C^*$-algebra $\caW_{\al}$ be generated by $S \in \caB(\scrH_\sys)$ and $\caW(\psi), \psi \in \frh_{\al}$.  For any $O \in \caW_{\al}$, we have
\beq 
\lim_{t  \to \infty}  \langle \Psi_t, O\Psi_t \rangle =    \langle \Psi_{\gs},  O \Psi_{\gs} \rangle 
\eeq
\end{lemma}

We will also need the time-independent bound on the total number of photons.

\begin{lemma}[Photon bound \cite{deroeckkupiainenphotonbound}]\label{lem: bound on number operator}
There is a $\ka>0$ such that 
\beq 
\sup_{t }  \langle \Psi_t, \e^{\ka N} \Psi_t \rangle  \leq  C <\infty
\eeq
with $C$ depending on $\Psi_0 \in\caD_\al$ and $\ka$.
\end{lemma}

The following result describes the localization of photons in the
ground state.

\begin{lemma}[Localization of ground state photons] \label{lem: localization ground state}
Let $\theta$ be a smooth indicator such that $ \supp \theta  \subset \{ 0< \str x \str\}$. Then 
\beq 
  \langle\Psi_{\gs}, \d \Ga (\theta_t) \Psi_{\gs}\rangle =\caO( \langle t \rangle ^{-\al})
\eeq
\end{lemma}

The next lemma bounds the number of photons after time $t$ in spatial regions around the atom. 

\begin{lemma}[Propagation bound \cite{deroeckkupiainenpropagation}] \label{lem: propagation bound}
Let $\theta$ be a smooth indicator such that $ \supp \theta  \subset \{  \str x \str < 1 \} $. Then for any time $t_c\geq \str\la\str^{-2}$ 
the limit $a(\theta, t_c):=\lim_{t\to\infty} \langle \Psi_t, \d \Ga (\theta_{t_c})\Psi_t \rangle$ exists and moreover
\beq  
  \langle \Psi_t, \d \Ga (\theta_{t_c})\Psi_t \rangle = a(\theta, t_c)  + \caO( \langle t \rangle ^{-\al})
\eeq
uniformly for $t_c \in [\str\la\str^{-2},t]$. 
\end{lemma}

In Section \ref{sec: strong local relaxation}, we will identify
\beq \label{eq: identification a}
a(\theta, t) = \langle  \Psi_{\gs}, \d \Ga (\theta_t) \Psi_{\gs}\rangle
\eeq
Note that in the case where $0 \not \in \supp \theta$, we can combine this propagation bound and \eqref{eq: identification a} with Lemma \ref{lem: localization ground state} to conclude that 
\beq \label{eq: field in transition region}
  \langle \Psi_t, \d \Ga (\theta_t)\Psi_t \rangle =  \caO( \langle t \rangle ^{-\al})
\eeq

Finally, we need a more detailed photon bound
\begin{lemma}[Soft photon bound \cite{deroeckkupiainenpropagation}]\label{lem: control on small momenta}
For any $\ep>0$, we have
\beq 
\sup_{t }  \langle \Psi_t, \d \Ga (\lone_{\str k \str \leq \ep})\Psi_t \rangle  \leq  C\ep^{\al/2} 
\eeq
with $C$ depending on $\Psi_0 \in \caD_\al$ and $\al$, but not on $\ep$.
\end{lemma}


\section{Strong local relaxation} \label{sec: strong local relaxation}

In Lemma \ref{lem: weak relaxation}, we saw that the system relaxes to the ground state in a weak sense.
The following result expresses that, locally  in space, it also relaxes  in a strong sense.
This will also allow us to identify the numbers $a(\theta, t)$ in Lemma \ref{lem: propagation bound}.

Let $B(r)\subset \bbR^3$ be the ball with radius $r$ in position space and $B^c(r)= \bbR^3 \setminus B(r)$. We will often split 
\beq
\scrH = \scrH_{B(r)} \otimes  \scrH_{B^c(r)}
\eeq
where it is understood that the small system is incorporated in the inner sphere, i.e.
\beq
 \scrH_{B(r)} := \scrH_\sys \otimes \Ga (L^2 (B(r))), \qquad   \scrH_{B^c(r)} := \Ga (L^2 (B^c(r)))
\eeq
Let us then write in general $\Tr_{B(r)}$ and $\Tr_{B^c(r)}$ for the
partial traces over $\scrH_{B(r)}$ and $\scrH_{B^c(r)} $, respectively. 

\begin{proposition}[Strong local relaxation] \label{prop: fate of atom}
Fix $r>0$. Then 
\beq 
 \lim_{t \to \infty}   \Tr_{B^c(r)}  \str \Psi_t \rangle \langle \Psi_t \str  =    \Tr_{B^c(r)}     \str \Psi_{\gs} \rangle \langle \Psi_{\gs} \str   
 \eeq
\end{proposition} 

\begin{proof}
Let $\Delta_{B(r)}$ be the Laplacian on the ball $B(r)$ with Neumann boundary conditions.  Then, for any $\psi$ in its form domain,
\beq
- \langle \psi, \Delta_{B(r)} \psi \rangle^{}_{L^2(B(r))}= \int_{B(r)}\d x \str\nabla \psi(x)\str^2
\eeq
and this form domain is in fact  the Sobolev space  $W^{2,1}(B(r))$
(completion of $C^{\infty}(B(r))$ w.r.t.\ the norm $\norm \psi \norm+\norm  \nabla \psi \norm$).  Hence, in the sense of quadratic forms on the form domain of $ \Delta_{\bbR^d}$, 
\beq
- \lone_{B(r)}   \Delta_{B(r)}  \lone_{B(r)} \leq  - \Delta_{\bbR^d}.
\eeq
For $\Psi \in \Dom(H_\res)$, we then get 
\beq
\langle \Psi, H^2_\res \Psi \rangle \geq  \langle \Psi, \d\Ga(-\Delta_{\bbR^d}) \Psi \rangle  \geq  \langle \Psi, \d\Ga(-\lone_{B(r)}\Delta_{B(r)}  \lone_{B(r)}) \Psi \rangle
\eeq
Using the relative boundedness of $H_\inter$ w.r.t\ to $H_\res$, the boundedness of $H_\sys$ and the fact that $H^2$ is conserved,  we obtain
\beq
\sup_t\langle \Psi_t, H_\res^2 \Psi_t \rangle \leq    C \norm \Psi_0 \norm^2+ \norm H \Psi_0\norm^2.
\eeq
Let $N_{B(r)}=\d\Ga(1_{B(r)})$, then by the photon number bound we
have 
\beq 
    \sup_t \langle \Psi_t, N^2_{B(r)} \Psi_t \rangle \leq C. 
\eeq  
We can restate these two a priori bounds by saying that the orbit of the reduced density matrix 
\beq
     \rho_{t,r} := \Tr_{B^c(r)} \str \Psi_t \rangle \langle \Psi_t \str
\eeq
remains in 
\beq
    \caK= \caK_{C_N,C_{\Delta}} := \left\{ \rho \in \caB_1(\scrH_{B(r)})\, \str \,  \Tr[  N_B^2 \rho ] < C_{N},   \Tr [ \d\Gamma(-\Delta_{B(r)})  \rho]  < C_{\Delta} \right\}
\eeq
for some fixed $C_N, C_\Delta$ (that depend on the initial condition).   By the compact Sobolev embedding, we easily find that the set $\caK$ is compact in the trace norm topology,  and hence all sequences  $\rho_{t_j,r}$ with $t_j \to \infty, j \to \infty$ have convergent subsequences.   To argue that they necessarily converge to  $\rho_{\infty, r}:= \Tr_{B^c(r)}    \str \Psi_{\gs} \rangle \langle \Psi_{\gs} \str $, we need the Lemma~\ref{two-densities} below.
We now use this Lemma with $\rho$, a limit point of $ \rho_{t,2r}$, and with $\rho'=  \Tr_{B^c(2r)}  \str \Psi_{\gs} \rangle \langle \Psi_{\gs} \str  $. 
The hypothesis $\rho_r \neq \rho'_{r}$ leads to a contradiction because of the weak relaxation (Lemma \ref{lem: weak relaxation}).
\end{proof}

\begin{lemma} \label{two-densities}
Let $\rho,\rho' \in \caB_1(\scrH_{B(2r)})$ be positive operators satisfying $\Tr \rho N^2, \Tr \rho' N^2  <\infty$, and denote by $\rho_r, \rho'_r$ their respective restrictions (partial traces) to $\caB_1(\scrH_{B(r)})$.   If  $\rho_{r} \neq \rho'_r$, then there is a $W \in \caW_\al \cap \caB(\scrH_{B(2r)})$ such that 
\beq
\Tr \rho W   \neq  \Tr \rho' W
\eeq
\end{lemma}

\begin{proof}
Given a dense set $\frg \subset L^2(B(r))$, the $C^\star$ algebra generated by the elements $S \in \caB(\scrH_\sys)$ and $\caW(\psi), \psi \in \frg$ separates points in $\caB_1(\scrH_{B(r)})$.  This does not immediately prove the lemma because $ \frh_\al \cap L^2(B(r)) $ need not be dense in $L^2(B(r))$. 
However, let us  consider the dense set  $ 1_{B(r)} \frh_\al $ and write $\psi=\psi_{B(r)} +\psi_{B^c(r)}$ with $\psi_{B(r)}= 1_{B(r)} \psi_{B(r)}$. 
Since  (see, for example, \cite{derezinskigerardmassive}) 
\beq 
\norm (\caW(\psi)-\lone)(1+N)^{-1} \norm \leq  C  \norm \psi \norm_{\frh} 
\eeq
and $\caW(\psi)= \caW(\psi_{B(r)}) \caW(\psi_{B^c(r)}) $, we get 
\beq
\str \Tr \rho \caW(\psi) -   ´\Tr \rho \caW(\psi_{B(r)})  \str  \leq   C  \norm \psi_{B^c(r)} \norm^2   \norm \caW(\psi_{B(r)})  \norm \norm \rho N \norm_1
\eeq
and analogously for $\rho'$. This proves the claim because, given an element  $ \varphi \in  1_{B(r)} \frh_\al $ and $\ep >0$,  we can find $\psi  \in \frh_\al$ such that that $\varphi= \psi_{B(r)} $ and   $\norm \psi_{B^c(r)} \norm \leq \ep $. 
\end{proof}

By combining the strong relaxation with the photon bound, it is clear that we also get relaxation of some unbounded observables on $\scrH_{B(r)}, r <\infty$. Indeed, let $O$ be an operator on $\scrH_{B(r)}$ such that, say  $\norm O \lone_{N=n}\norm$ grows not faster than polynomially in $n$, then the second term in 
\beq
\langle \Psi_t,  O \Psi_t \rangle =   \langle \Psi_t,  O \lone_{N \leq n} \Psi_t \rangle +   \langle \Psi_t,  O \lone_{N > n} \Psi_t \rangle \\
\eeq
can be bounded uniformly in time $t$ by a number that vanishes as $n \to \infty$, and the first term converges as $t \to \infty$.  It follows that
\beq \label{eq: relaxation unbounded}
\lim_{t\to \infty} \langle \Psi_t,  O \Psi_t \rangle =   \langle \Psi_\gs,  O \Psi_\gs \rangle
\eeq
In particular, this allows us to make the identification  \eqref{eq: identification a} thus supplementing Lemma \ref{lem: propagation bound}.
To end this section, we  give the 

\subsubsection*{Proof of Lemma \ref{lem: exp phonon loc}} \label{sec: proof of loc lemma}
 If Lemma \ref{lem: exp phonon loc} were false, then for any $\ka>0$,  $ \limsup_{n \to \infty} a(n)\e^{\ka n}= \infty$, with $a(n):= \norm \lone_{N=n}\Psi_{\gs} \norm$. 
For any number $m$, we could then find  $n \in \bbN$ and $r <\infty$ such that 
\beq  \label{eq: contra one}
 \norm \lone_{N=n} \Ga(1_{B(r)}(x)) \Psi_{\gs} \norm \geq m \e^{-\ka n} 
\eeq
Let $\ka >0$ be as in Lemma \ref{lem: bound on number operator}  and choose $m$ such that 
\beq  \label{eq: contra two}
\sup_t \norm \lone_{N=n} \Psi_{t} \norm \leq (m/2) \e^{-\ka n} 
\eeq
for some arbitrarily chosen $\Psi_0$.    Since $O= \lone_{N=n} \Ga(1_{B(r)}(x)) $ is a bounded operator on $\scrH_{B(r)}$, Proposition \ref{prop: fate of atom} implies that  $\langle \Psi_t, O \Psi_t \rangle \to  \langle \Psi_{\gs}, O \Psi_{\gs} \rangle $, which is clearly incompatible with (\ref{eq: contra one}, \ref{eq: contra two}).


\section{Proof of the main result} \label{sec: scattering setup}

\subsection{Scattering identification and inverse}

We introduce some standard tools used in scattering theory. In their present form, these tools were introduced in  \cite{derezinskigerardmassive} and we follow very closely, if not literally, this source.  First, we recall the natural isomorphism of Hilbert spaces
\beq
\Ga(\frh\oplus \frh) \sim   \Ga( \frh)  \otimes \Ga(\frh) 
\eeq
Consider $j_0, j_\infty$ bounded operators on $\frh$ satisfying 
\beq \label{eq: conditions on j}
j_0+j_\infty=1 , \qquad   j_0^* j_0+ j_{\infty}^*j_\infty \leq \lone
\eeq
and define $j : \frh \to \frh \oplus \frh:  \phi \mapsto   (j_0 \phi, j_\infty \phi)$, satisfying $\norm j \norm \leq 1$ because of \eqref{eq: conditions on j}.
We construct the operators $\breve\Ga(j) : \Ga(\frh) \to \Ga(\frh) \otimes   \Ga(\frh)$  by 
\begin{align}
 & \breve\Ga(j)  a^{\ast}(f_1) a^{\ast}(f_2)  \ldots  a^{\ast}(f_m)  \Om    \\[2mm]
 & :=        \left(  a^{\ast}( j_0 f_1) \otimes \lone + \lone 
 \otimes a^{\ast}( j_\infty f_1)    \right)        \ldots \left(  a^{\ast}( j_0 f_m) \otimes \lone + \lone 
 \otimes a^{\ast}( j_\infty f_m)    \right)(\Om\otimes \Om)  
\end{align}
where $\Om$ are Fock vacua.   Note that $\norm\breve\Ga(j) \norm \leq 1 $.  
Hence, the picture that goes with the above definition of
$\breve\Ga(j) $ is that each photon is split into a part close to the
origin and one part far from the origin. The former is put in the first copy of Fock space, the latter into the second.

Let us now pick a particular
 $j=j(t): \frh \to \frh\oplus \frh$ defined by $j_0(t) = \theta(\cdot/t) $ and $j_\infty(t)=1-j_0(t)$ with $\theta$ a smooth indicator (see Section \ref{sec: technical tools}) such that $\supp \theta \subset \{\str x \str \leq \largev\}$ and $\theta=1$ on $ \{\str x \str \leq \smallv\}$, with $\smallv < \largev <1$.  It is clear that  $j(t)$ satisfies the conditions \eqref{eq: conditions on j}. 
With this choice of $j(t)$, we construct an 'extended inverse scattering identification':
\beq
J_{\extension, t}  :=  \lone \otimes \breve\Ga (j(t)) :  \scrH \to  \scrH_{as}
\eeq
where $\scrH_{\mathrm{as}}$ is the asymptotic space defined by
\beq
\scrH_{\mathrm{as}} = \scrH \otimes \scrH_\res \sim \scrH_\sys \otimes \scrH_\res \otimes \scrH_\res 
\eeq
 
The operator $J_{\extension, t}$  is a one-sided inverse to the
'extended scattering identification' $I_\extension: \scrH_{\mathrm{as}} \to \scrH$, see \eqref{eq: inverse indeed}. To construct $I_\extension$, define it first as a closable operator $ \Ga(\frh \oplus \frh) \sim \scrH_\res \otimes \scrH_\res \to \scrH_\res$ by  $ \Ga (\iota)$ with $\iota: \frh\oplus \frh \to \frh: (\phi, \phi')\mapsto \phi+\phi'$, or, in a more intuitive representation, by 
\beq 
I_{\extension}(a^*({f_1}) \ldots a^*({f_n})  \Omega\otimes  a^*({f'_1}) \ldots  a^*({f'_{n'} })  \Omega)  =  a^*({f_1}) \ldots  a^*({f_n}) a^*({f'_1}) \ldots  a^*({f'_{n'}})  \Omega
\eeq
and then extend by tensoring with $\lone$ on $\scrH_\sys$. 
 Note that $\iota j(t)=\lone$ and hence indeed
\beq  \label{eq: inverse indeed}
 I_{\extension} J_{\extension, t} =\lone.
\eeq
The operator $I_{\extension}$ is unbounded, but by using $\norm a^{}(f) \lone_{N =n}\norm = \sqrt{n}\norm f \norm_\frh $, one can compute
\beq \label{eq: crude bound on extension i}
\norm I_{\extension} (\lone_{N < n'} \otimes \lone_{N<n}) \norm^2 = \frac{(n+n')!}{n! n'!}
\eeq
$I_{\extension} $ is related to $I$ by 
\beq 
 I_{\extension} \Psi_{\gs}\otimes \Psi_\res =  I  \Psi_\res  
\eeq
For our analysis, it is important that $I \lone_{N \leq n}$ is bounded for any $n$, as follows by combining
 \eqref{eq: crude bound on extension i} and Lemma \ref{lem: exp phonon
   loc}.  
In particular, this implies that the limit $W_t \Psi$ in \eqref{eq: limit of w operator} exists for \emph{any} $\Psi \in \caD_{fin}$.
Finally, we also define
\beq 
J_t :=    \langle \Psi_{\gs} \str \otimes \lone  \, J_{\extension, t} :   \scrH \to \scrH_\res
\eeq
and we will argue that $J_t$ is an approximative inverse to $I$.

\subsection{Proof of main result} \label{sec: proof of main result}

We  define the asymptotic Hamiltonian $H_{\mathrm{as}}$  on $\scrH_{\mathrm{as}}$ as
\beq \label{eq: as ham}
H_{\mathrm{as}} =  H \otimes \lone + \lone \otimes H_\res
\eeq
The following two propositions are our  main technical results.
To state these results, and also in the remainder of the paper, we assume that Assumptions \ref{ass: infrared behaviour} and \ref{ass: fermi golden rule} hold and that  $\str\la\str$ is sufficiently small.  This will not be repeated any more.

\begin{proposition}[Range of inverse wave operator] \label{prop: range wave}
For $\Psi \in \caD_\al$,
\beq
 \lim_{n \to \infty} \limsup_{t \to + \infty}  \big\norm \big(I \lone_{[N \leq n]} J_t -\lone \big) \e^{-\i t H}\Psi \big\norm =  0.
\eeq

\end{proposition}

\begin{proposition}[Existence of inverse wave operator] \label{prop: existence of inverted}
The limit
\beq
   Z_{\extension}\Psi := \lim_{t \to + \infty}   \e^{\i t H_{\mathrm{as}}} J_{\extension, t}  \e^{-\i t H}\Psi
\eeq
exists for $\Psi \in \caD_\al$. Note that $Z_\extension$ extends to a contraction $\scrH\to \scrH_\mathrm{as}$ because $J_\extension$ is a contraction.
\end{proposition}

The remaining sections of the paper are devoted to the proof of Propositions \ref{prop: range wave} and \ref{prop: existence of inverted}. We now show how they lead to  our main result. 

\begin{proof}[Proof of Theorem \ref{thm: ac}]
Define the contraction $Z: \scrH\to \scrH_\res$ by 
\beq  \label{eq: lim of z}
Z\Psi := \langle \Psi_\gs \str \otimes \lone \,  (Z_{\extension} \Psi) = \lim_{t \to +\infty}   \e^{\i t (E_\gs+H_\res)}J_t \e^{-\i t H}\Psi.
\eeq
Using this and the fact that $\lim_{t \to \infty} W_t \Psi=W_+\Psi$ holds for any $\Psi \in \caD_{fin}$ (Theorem \ref{thm: w op} combined with the remark following \eqref{eq: crude bound on extension i}), we get 
\baq
   \lim_{t \to +\infty}  \e^{\i t H} I  \e^{-\i t (E_\gs+H_\res)} \lone_{[N \leq n]} \e^{\i t (E_\gs+H_\res)} J_t  \e^{-\i t H}\Psi =   W_+  \lone_{[N \leq n]} Z\Psi
\eaq
Since $W_+$ is an isometry, we can take the limit $n \to \infty$ on the RHS and obtain $W_+ Z\Psi$. 
By Proposition \ref{prop: range wave}, the $n \to \infty$ limit of the LHS equals $\Psi$.  
Since $ \caD_\al$ is dense,  we have proven that $W_+ Z$ extends to the identity on $\scrH$, hence in particular $\Ran W_+=\scrH$. 
\end{proof}

\section{Range of the wave operator}
In this section, we prove Proposition \ref{prop: range wave}. 
Fix a smooth indicator $\theta^{(1)}$ such as described in Section \ref{sec: technical tools}. Let  $0<v_1<,v_2<1$ be such that
 $\theta^{(1)}(x)=1$ for $\str x \str \leq \smallv$ and $0$ for $\str x \str \geq \largev$.  The indicator $\theta^{(1)}$ is used to  construct the identification operator $J_t$ (see Section \ref{sec: scattering setup}). Moreover, we need another 
 smooth indicator $\theta^{(2)}$ such that $ \theta^{(2)}(x) =1$ for $\str x \str \leq  v_2 $ but still $\supp \theta^{(2)} \subset B(v_3)$, for some $v_3 <1$. 
The relevant distances are summarized in Figure \ref{fig: scalesonaxis}.
\begin{figure}[h!] 
\vspace{0.5cm}
\begin{center}
\def\svgwidth{\columnwidth}
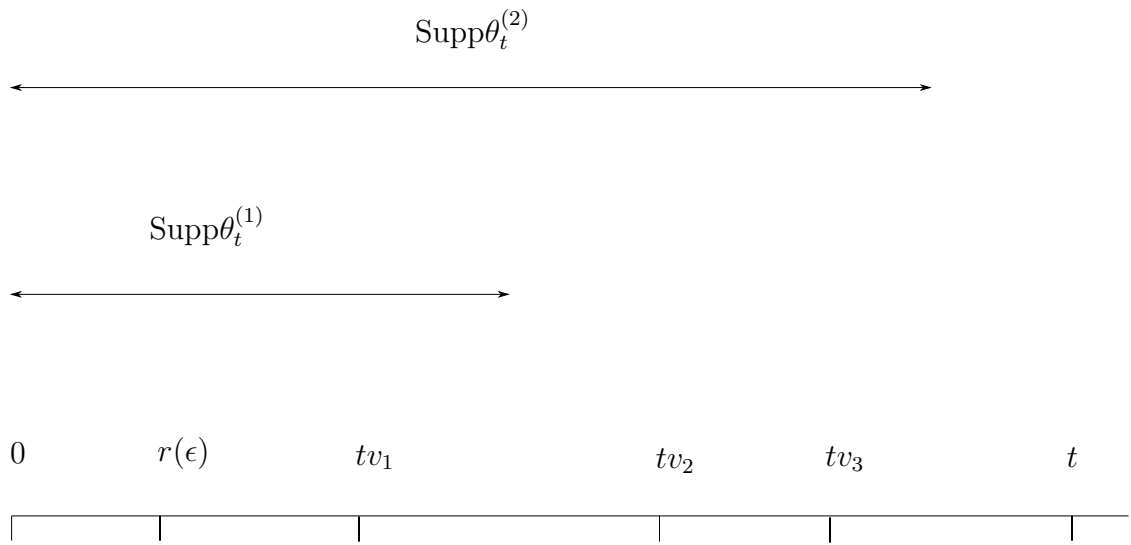
\caption{ \footnotesize{Important distances from the origin to the position of the outgoing waves, at $r \sim t$.}   \label{fig: scalesonaxis} }
\end{center}
\end{figure}
Note that for  a measurable $K \subset \bbR^d \setminus \{0\}$,  the space $\scrH_{K}$ is Fock space in a natural way, namely
\beq
\scrH_{K} = \Ga( L^2(K))
\eeq
and we denote the corresponding vacuum by $\Om_{K}$.   We also write $N_K= \d \Ga(1_K)$ for the photon number in the region $K$. 

We define 
\beq
\Psi_{\gs,B(r)} :=  \langle \Om_{B^c(r)} \str \Psi_{\gs} \in \scrH_{B(r)}.
\eeq
Then  $\Psi_{\gs,B(r)} \otimes \Om_{B^c(r)}$ is an approximation to the ground state, more precisely
\beq \label{eq: convergence of gs}
  \norm  \Psi_{\gs,B(r)} \otimes \Om_{B^c(r)} -   \Psi_{\gs} \norm  \to 0, \qquad  r \to \infty
\eeq


\subsubsection*{Quantitities depending on $\ep>0$} 
We fix  a small $\ep>0$ and we now choose  $r(\ep), t(\ep)$ as follows:
We first choose $r(\ep)$ such that for $r >r(\ep)$,
\beq \label{ec: first cond on eps}  
\big|\norm \langle \Psi_{\gs,B(r)} \str \Psi_\gs \norm^2-1 \big| \leq \ep, \qquad \langle \Psi_\gs, N_{B^c(r)} \Psi_\gs \rangle \leq \ep,
\eeq
then we choose $t(\ep)$ such that for $t>t(\ep)$ and $r=r(\ep)$:
\begin{enumerate}
\item  
\beq
\big| \norm \langle \Psi_{\gs,B(r)} \str \Psi_t \norm^2-1 \big| \leq \ep, \qquad  \left\str \langle \Psi_\gs, N_{B(r)} \Psi_\gs  \rangle-  \langle \Psi_t, N_{B(r)} \Psi_t  \rangle \right\str \leq \ep.
\eeq
This is possible by  Proposition \ref{prop: fate of atom} and the remarks following it. 
\item 
The error term $\caO(t^{-\al})$ in Lemma \ref{lem: propagation bound}
  with $\theta=\theta^{(1)}$ and $\theta=\theta^{(2)}$ is not larger than $\ep$.
\end{enumerate}
 
For $r=r(\ep)$ and $t\geq t(\ep)$, as introduced above, we define the time-dependent 
annular region
\beq
M_t= B(t v_2) \backslash B(r).
\eeq
Our proof of  Proposition \ref{prop: range wave} exploits that there are no photons in $M$ (up to  $\caO(\ep)$-terms).  
Indeed we have:
\begin{lemma} \label{lem: close to one dim}
For $r=r(\ep)$ and $t\geq t(\ep)$:
\beq
\norm \langle \Psi_{\gs,B(r)}\otimes \Om_{M_t}\str \Psi_t  \norm = 1+\caO(\ep)
\eeq
\end{lemma}

\begin{proof}
From $1_{B(r)} + 1_M \leq \theta_t^{(2)}$ we get
\beq
   \langle \Psi_t,  N_{M_t} \Psi_t \rangle +  \langle\Psi_t,  N_{B(r)}\Psi_t \rangle \leq  \langle \Psi_t, \d\Ga( \theta_t^{(2)}) \Psi_t \rangle.
\eeq
In $\langle\Psi_t,  N_{B(r)}\Psi_t \rangle$ and  $\langle \Psi_t, \d\Ga( \theta_t^{(2)}) \Psi_t \rangle$ 
we now replace the vector $\Psi_t$ by $\Psi_\gs$ using our choices for $r$ and $t(\ep)$ expressed
in Points 1. and 2. above. Then, using $\d\Ga( \theta_t^{(2)}) \leq N_{B(r)} + N_{B^c(r)}$ as well as \eqref{ec: first cond on eps} we arrive at 
\beq
    \langle \Psi_t,  N_{M_t} \Psi_t \rangle \leq \langle \Psi_\gs,  N_{B^c(r)} \Psi_\gs \rangle + 2\ep \leq 3 \ep,
\eeq
Let us now  define the projectors on $\scrH$
 \beq P_1= \frac{\str \Psi_{\gs,B(r)}\rangle \langle \Psi_{\gs,B(r)}\str}{\norm \Psi_{\gs,B(r)} \norm^2} \otimes \lone,\qquad   P_2= \lone \otimes \str \Om_{M_t} \rangle \langle \Om_{M_t}\str \otimes \lone\eeq
and write $P^\perp=\lone-P$ for projectors $P$, then we get (for $j=1$ from \eqref{ec: first cond on eps}   and for $j=2$ from the above)
\beq
   \langle \Psi_t, P^\perp_j \Psi_t \rangle \leq C \ep
\eeq
and therefore, since $(P_1P_2)^\perp \leq P_1^\perp+P_2^\perp $, 
\beq
\langle \Psi_t, (P_1P_2)^\perp \Psi_t \rangle \leq C\ep
\eeq
which yields the lemma.
\end{proof}

\begin{lemma}\label{lem: qit}
Consider Hilbert spaces $\scrH_a, \scrH_b$. 
Let $\Psi \in \scrH_{a} \otimes \scrH_b$, $\norm \Psi\norm  =1$ and let $\Psi_a \in \scrH_a, \Psi_b \in \scrH_b$ be such that
\beq
\Psi_b= \langle \Psi_a \str \Psi, \qquad \str 1- \norm \Psi_a \norm^2\str,   \str 1-\norm \Psi_b \norm^2 \str \leq \nu.
\eeq
Then $\norm \Psi-\Psi_a \otimes \Psi_b \norm^2 \leq C \nu$.
\end{lemma}

\begin{proof}
By the definition of $\Psi_b$,  $\langle \Psi, \Psi_a \otimes \Psi_b \rangle =\norm \Psi_b \norm^2 $ and hence
\beq
\langle \Psi-\Psi_a \otimes \Psi_b, \Psi-\Psi_a \otimes \Psi_b \rangle= 1+\norm \Psi_a \norm^2\norm \Psi_b \norm^2 -2 \norm \Psi_b \norm^2
\eeq
\end{proof}

\subsubsection*{Proof of Proposition \ref{prop: range wave}}
We will apply Lemma \ref{lem: qit} with
$$\Psi=\Psi_t, \qquad \Psi_a= \Psi_{\gs,B(r)}\otimes
\Om_{M_t}, \qquad \Psi_{b}= \langle \Psi_a \str \Psi_t = : \Psi_{B^c(tv_2),t}.$$
To recall our conventions, note that  $B(tv_2)=M_t \cup B(r)$ and  $$  \Psi_{\gs,B(r)}\otimes
\Om_{M_t}\in \scrH_{B(r)}\otimes \scrH_{M_t} =  \scrH_{B(tv_2)}, \qquad  \Psi_{B^c(tv_2),t} \in \scrH_{B^c(tv_2)}.$$ 
By Lemma \ref{lem: close to one dim}, we can take $\nu=C\ep$ and then Lemma \ref{lem: qit}  yields
\beq \label{eq: 1}
\norm \Psi_t-  \Psi_{\gs,B(r)}\otimes \Om_{M_t} \otimes \Psi_{B^c(tv_2),t} \norm  \leq  C \sqrt{ \ep}
\eeq
By the uniform in time photon bound, we have
\beq \label{eq: photon bound outer}
\sup_{t} \langle \Psi_{B^c(tv_2),t}, N \Psi_{B^c(tv_2),t} \rangle \leq C.
\eeq
Since the region where $\theta_t^{(1)}$ is nonconstant is in $B(t v_2) \setminus B(t v_1) \subset M_t$, we have
\beq
 J_{\extension, t}  \Psi_{\gs,B(r)}\otimes \Om_{M_t} \otimes  \Psi_{B^c(tv_2),t}= (\Psi_{\gs,B(r)}\otimes \Om_{B^c(r)}) \otimes (\Om_{B(t v_2)} \otimes \Psi_{B^c(tv_2),t}). 
 \eeq  
Therefore, by \eqref{eq: 1} and the choice of $r$, i.e.\ \eqref{ec: first cond on eps},
\beq \label{eq: conclusion}
\sup_{t\geq t(\ep)} \norm J_t \Psi_t -     \Om_{B(t v_2)} \otimes   \Psi_{B^c(tv_2),t} \norm \leq C \sqrt{ \ep}
\eeq
(keep in mind that we have constrained $t \geq t(\ep)$ throughout this section). 
Next, let us consider the map
\beq
R_{n,r'}:  \lone_{N \leq n} \scrH_{B^c(r')} \to \scrH:  \Psi \mapsto I(\Om_{B(r')}\otimes \Psi) - \Psi_{\gs, B(r')}\otimes \Psi 
\eeq
Then, for any $n$, we have $\lim_{r' \to \infty} \norm R_{n,r'} \norm=0$ by Lemma \ref{lem: exp phonon loc} and the bound \eqref{eq: crude bound on extension i}. Since $\sup_t \norm \Psi_{B^c(tv_2),t} \norm \leq C$, this yields
\beq 
 \lim_{t \to +\infty}\norm I  (  \Om_{B(t v_2)} \otimes \lone_{N \leq n} \Psi_{B^c(tv_2),t}) -  \Psi_{\gs, B(t v_2)}\otimes \lone_{N \leq n}\Psi_{B^c(tv_2),t}  \norm  =0
\eeq
Combining this with \eqref{eq: conclusion} and the fact that $\norm I \lone_{N \leq n} \norm <\infty$, we get 
\beq \label{eq: identitiy with cutoff}
 \lim_{t \to +\infty} \norm I  \lone_{N \leq n}J_t \Psi_t  -  \Psi_{\gs, B(t v_2)}\otimes \lone_{N \leq n} \Psi_{B^c(tv_2),t} \norm \leq C \sqrt{ \ep}
\eeq 
Next, we argue that
\beq \label{eq: relax nb cutoff}
\lim_{n \to +\infty}\limsup_{t \to +\infty} \norm \Psi_t - \Psi_{\gs, B(t v_2)}\otimes \lone_{N \leq n} \Psi_{B^c(tv_2),t}  \norm \leq C \sqrt{ \ep}
\eeq
Indeed, by our choice of $r(\ep)$, i.e.\ \eqref{ec: first cond on eps}, we have
\beq
\sup_{t \geq t(\ep)}\norm \Psi_{\gs, B(tv_2)}- \Psi_{\gs,B(r)}\otimes \Om_{M_t}\norm  \leq C \ep
\eeq
and by the photon bound \eqref{eq: photon bound outer} we have  $ \lim_{n\to +\infty}\sup_t\norm \lone_{N>  n} \Psi_{B^c(tv_2),t}  \norm = 0$.  Therefore,  \eqref{eq: relax nb cutoff} follows from  \eqref{eq: 1}. 

Since $\ep$ can be chosen arbitrarily small, the limits  \eqref{eq: identitiy with cutoff} and \eqref{eq: relax nb cutoff} equal $0$. Together, they imply Proposition \ref{prop: range wave} by the triangle inequality.

\newpage
\section{Existence of wave operator}\label{sec: existence of wave operator}

In this section, we prove Proposition \ref{prop: existence of inverted}. 
We calculate  (the formal manipulations are justified since $ \norm N \Psi_t \norm <\infty$)
\beq \label{eq: cov derivative one}
 \i \frac{\d}{\d t}  \e^{\i t H_{\mathrm{as}}} J_t  \e^{-\i t H}\Psi =   e^{\i t H_{\mathrm{as}}} ( \lone_\sys \otimes B_t)   \Psi_t+    \e^{\i t H_{\mathrm{as}}}   \left(  (1 \otimes  H_I)     \breve\Ga(j)  -  \breve\Ga(j) H_I  \right) \Psi_t 
\eeq
where 
\begin{align}
B_t=  ( \d \Ga(\om) \otimes 1+ 1 \otimes     \d \Ga(\om))     \breve\Ga(j)  -  \breve\Ga(j)  \d\Ga(\om) -\i \frac{\d}{\d t} \breve \Ga(j)
\end{align}
To estimate the second term in \eqref{eq: cov derivative one}, we remark that it is bounded by 
\beq
C\norm a^{\#}(j_{\infty}\phi) \Psi_t  \norm  \leq   C  \norm j_{\infty} \phi \norm \norm N^{1/2} \Psi_t \norm \leq  C  \norm j_{\infty} \phi \norm
\eeq
where we used the photon bound.  By Assumption \ref{ass: infrared behaviour} and the remark following it, we can deduce that  $ \norm j_{\infty} \phi \norm \leq \langle t \rangle^{-(1+\ga)}$ for some $\ga>0$, and therefore this term is integrable.
We focus on the first term in \eqref{eq: cov derivative one}.  From  Lemma 2.16. in \cite{derezinskigerardmassive}, we have the bound
\beq
\norm B_t \Psi_t \norm \leq   2 \norm N^{1/2} \, \Psi_t \norm  \langle \Psi_t,  \d \Ga (w_0^{*}w_0
)  \Psi_t  \rangle^{1/2}
\eeq
where
\beq
w_0= \i [\om, j_0] - \frac{\d}{\d t} j_0.
\eeq
Because of the photon bound in Lemma \ref{lem: bound on number operator}, 
Proposition \ref{prop: existence of inverted} will be proven if we show that
\beq
\langle \Psi_t,  \d \Ga (w_0^{*}w_0) 
 \Psi_t  \rangle^{1/2} \leq C \langle t \rangle^{-(1+\ga)}
\eeq
for some $\ga>0$. In the following the Lemma~\ref{lem: bound on
  dgamma}, below, is used repeatedly.


\subsection{Bound on  $ \langle \Psi_t, \d \Ga (w_0^* w_0)  \Psi_t  \rangle^{1/2}$}

We use the results of Appendix C to write
\begin{align} \label{eq: decomposition of w_0}
w_0 =   \frac{1}{t}  \left(\frac{k}{\str k \str}- \frac{x}{t}\right)\cdot  \nabla \theta( x /t) + \langle t\rangle^{-1-\be} b_t
+ \langle t\rangle^{-1}  b'_t \lone_{[  \str k \str \leq t^{\be-1}]  }
\end{align}
where  the exponent $\be$ satisfies $0<\be<1$  and $b_t,b'_t$ are operators that are bounded uniformly in $t$.  
By using Cauchy-Schwartz and  Lemma  \ref{lem: bound on dgamma}, it suffices to prove that 
\beq
\langle \Psi_t, \d \Ga (w_{0,i}^*w_{0,i}) \Psi_t  \rangle= \caO(t^{-2-\ga}), \qquad \ga>0
\eeq
where $w_{0,i}, i=1,2,3$ is any of the three terms on the RHS of \eqref{eq: decomposition of w_0}. We discuss the three terms separately. 

\subsubsection*{The first term}
Let
\beq
a_t= \left(\frac{k}{\str k \str}- \frac{x}{t}\right)\cdot  \nabla \theta( x /t)
\eeq
Since $\nabla \theta$ has support in $B(1)$ and  $\norm \left(\frac{k}{\str k \str}- \frac{x}{t}\right) 1_{[\str x/t \str\leq 1] } \norm \leq C$,  we can bound
\beq
a^*_ta_t  \leq  C       \str  \nabla \theta( x /t)\str^2
\eeq
Therefore, Lemma \ref{lem: bound on dgamma}  and linearity of $b \to \d \Ga(b)$ yield the bound
\beq \label{eq: something}
 \langle \Psi_t, \d \Ga ( a^*_ta_t ) \Psi_t  \rangle \leq C    \langle \Psi_t,  \d \Ga ( \tilde \theta(x/t)  \Psi_t  \rangle, \qquad  \tilde \theta (x):= \frac{  \str  \nabla \theta( x )\str^2}{\sup_{x \in B(1)} \str \nabla \theta(x) \str^2}
\eeq
Since $\tilde \theta$ is a smooth indicator whose support does not include $0$, we can apply  the bound \ref{eq: field in transition region} with  $\theta$ replaced by $\tilde \theta$ to get  the bound $\caO(t^{-\al})$ for \eqref{eq: something}. This settles the first term in \eqref{eq: decomposition of w_0}.

\subsubsection*{The second term}
By Lemma \ref{lem: bound on dgamma},
\beq
\langle \Psi_t,  \d \Ga (b_t^*b_t) \Psi_t  \rangle \leq \norm b_t^*b_t \norm \langle \Psi_t ,N \Psi_t  \rangle  \leq C
\eeq
where the last inequality follows by the uniform boundedness of $b_t$ and the photon bound, Lemma \ref{lem: bound on number operator}.  This suffices for the second term, since we have $\be>0$.

\subsubsection*{The third term}
Again Lemma \ref{lem: bound on dgamma} gives
\beq
\d\Ga (1_{[  \str k \str \leq t^{\be-1}]  }  (b_t')^*b_t' 1_{[  \str k \str \leq t^{\be-1}]  })  )  \leq  \norm  (b'_t)^*b_t' \norm \d\Ga(1_{[  \str k \str \leq t^{\be-1}]  }) 
\eeq
so we just have to bound 
\beq
\langle \Psi_t,  \d\Ga(1_{[  \str k \str \leq t^{\be-1}]  })  \Psi_t   \rangle
\eeq
We invoke Lemma \ref{lem: control on small momenta} to bound the last expression by $\caO(t^{\al(\be-1)/2})$. Since $0<\be<1$, this is a negative power of $t$.

\medskip
The following lemma has been used repeatedly.  Its elementary proof can be found e.g. in \cite{DG1999}.

\begin{lemma}\label{lem: bound on dgamma}
Let $a,b$ be bounded operators on $\frh$. 
\begin{itemize}
\item[$1)$] If  $a \leq b$, then 
\beq
\d \Ga (a)  \leq  \d \Ga (b)
\eeq
In particular $ \d \Ga (a) \leq  \norm a \norm N$ for any self-adjoint $a$. 
\item[$2)$]
\beq \label{eq: cs dgamma}
\str \langle \Psi\d \Ga (a b) \Psi\rangle \str  \leq \langle \Psi\d\Ga(aa^* )\Psi\rangle^{1/2} \langle\Psi\d\Ga (b^* b)\Psi\rangle^{1/2}
\eeq
\end{itemize}
\end{lemma}


\appendix
\renewcommand{\theequation}{\Alph{section}\arabic{equation}}
\setcounter{equation}{0}

%

%


\section{Commutator estimates}  
 \label{appendixC}
 
  \renewcommand{\theequation}{C-\arabic{equation}}
  \setcounter{equation}{0}  

 Let $f \in \caC_0^{\infty}(\bbR^d)$ and $f_t(x):= f(x/t)$.  Denote by 
  $f_t$ the corresponding multiplication operator  on $L^2(\bbR^d)$ and by $\om$ be the operator on $L^2(\bbR^d)$
 given by the Fourier multiplier $\om(k)=|k|$.  We want to study the commutator $[\om, f_t]$:
  \begin{lemma}  \label{lem: bound on commutator}
 Let $\nabla\om$ be the Fourier multiplier  $\nabla\om(k)=\hat k\equiv k/|k|$.  Then
 \beq\i [\om, f_t]
=  \nabla \om\cdot \nabla f_t + O_t
\eeq 
 where, for 
 $\be>0$ 
  \beq
  \norm O_t \lone_{[    t^{\be-1} \leq \str k \str ]} \norm \leq C t^{-1-\be}, \qquad  \norm O_t  \norm \leq C t^{-1
  }.   \label{eq: two commutator estimates}
 \eeq
\end{lemma}
\noindent {\bf Remark.} If  $\omega(k)$ were $C^2$  a standard result
 (e.g.\ Lemma 27 in \cite{frohlichgriesemerschleinrayleigh}) gives 
  \begin{align}
 \norm O_t  \norm \leq  t^{-2}   \norm \partial^2 \om \norm_{\infty}   \int \d q  \str \hat  f(q) \str  \str q\str^2 
 \end{align}
In our case $\str\partial^2 \om(p) \str \sim \frac{1}{\str p \str}$ and this result   is not applicable.

\begin{proof}
Let $\phi\in L^2(\bbR^d)$. We have
$$
\widehat{( [\om, f_t]\phi)} (k)=\int \d p\, \hat f_t(p)(\om(k)-\om(k-p))\hat\phi(k-p).
$$
Write
$$
\om(k)-\om(k-p)=|k|-|k-p|=p\cdot \hat k+r(k,p).
$$
We have, for  $|p|\leq |k|/2$
$$
|r(k,p)|=||k|-|k-p|-p\cdot \hat k|\leq Cp^2/|k|
$$ and for $|p|\geq |k|/2$,
$
|r(k,p)|\leq C|p|
$. Hence altogether
$$
|r(k,p)|\leq C\frac{p^2}{|p|+|k|}.
$$
We obtained
$$
i [\om, f_t]=\hat k \cdot\nabla f_t+R
$$
with
$$
| \widehat{(R\phi)} (k)|\leq C\int \d p\, |\hat f_t(p)|\frac{p^2}{|p|+|k|} |\hat\phi(k-p)|.
$$
Hence
$$
\|R\phi\|^2\leq \int \d p_1\d p_2 \d k \, 
\prod_{i=1}^2\hat f_t(p_i)|\frac{p_i^2}{|p_i|+|k|} |\hat\phi(k-p_i)|.
$$
Note that $\hat f_t(p)=t^d \hat f(tp)$ with $ \hat f$ rapidly decreasing. Thus, for $n \in \bbN$
 \beq\label{fmom}
\int \d p\, |\hat f_t(p)|p^n\leq C(n) t^{-n}.
\eeq
First, we use $\frac{p_i^2}{|p_i|+|k|}\leq |p_i|$ 
to bound
  \beq\label{smallk}
\|R\phi\|^2 \leq
 \int \d p_1\d p_2 \,\prod_{i=1}^2\hat p_if_t(p_i)|\int \d k\, 
|\hat\phi(k-p_1)||\hat\phi(k-p_2)|\leq
 C t^{-2}\|\phi\|^2
\eeq
where we used 
\eqref{fmom} with $n=1$ and Cauchy-Schwartz. This yields the second bound of \eqref{eq: two commutator estimates}.
For $\phi$ such that $\supp\ \hat\phi\subset \{|k|\geq1/t^{1-\beta}\}$,  we use $\frac{p_i^2}{|p_i|+|k|}\leq |p_i|^2/|k|$,  \eqref{fmom} with $n=2$ 
and $|k|^{-1}\leq t$, to  get
 \beq\label{largek}
\|R\phi\|^2 \leq
 C t^{-2-2\beta}\|\phi\|^2,
\eeq
which gives the first bound of \eqref{eq: two commutator estimates}.

\end{proof}

We also need
\begin{lemma}
Let 
 \beq
O_t:=  [ \hat k, f_t(x)] , 
\eeq 
Then 
for $\be>0$
  \beq
 \norm O_t \lone_{[    t^{\be-1} \leq \str k \str ]} \norm \leq C t^{-\be}, \qquad 
 \norm O_t   \norm \leq C.  
 \eeq
\end{lemma}
\begin{proof} 
We have
$$
\widehat{(O_t\phi)} (k)=\int \d p\hat f_t(p)q(k,p)\hat\phi(k-p).
$$
with
$$
|q(k,p)|=|\frac{ k}{|k|}-\frac{ k-p}{|k-p|}|\leq \frac{ |p|}{|p|+|k|}.
$$
Now proceed as in the previous lemma.
\end{proof}


%
%
%

\section{Localization of ground state photons}  
\label{app: virial argument}
 
\renewcommand{\theequation}{E-\arabic{equation}}
\setcounter{equation}{0}  

Here we prove Lemma \ref{lem: localization ground state}. By the
choice of $\theta_t$, we have for some $C>0$ and all $x\in\bbR^3$,
\beq \label{eq: bounding theta}
     0 \leq \theta_{t}(x) \leq C|t|^{-\alpha} |x|^\alpha. 
\eeq
It follows that 
\beq\label{eq:bounding theta2}
  \langle \Psi_\gs,  \d \Gamma(\theta_t) \Psi_\gs\rangle \leq
  C|t|^{-\alpha}  \int |x|^{\alpha} \|\Psi_\gs(x)\|^2 \d x
\eeq
where $\Psi_\gs(x)$ denotes the (inverse) Fourier transform of $a_k\Psi_\gs$.
Hence it remains to prove that \eqref{eq:bounding theta2} is finite. By a
well-know pull-through trick \cite{BFS1998,gerardgroundstate},
\beq
a(k)\Psi_\gs = \la \hat\phi(k) (H-E+\str k \str)^{-1} D  \Psi_\gs.        
\eeq
This implies, by  Assumption~\ref{ass: infrared behaviour} on $\phi$, that for some $\beta>\alpha$,
\begin{equation*}
     \|\partial_k^{m}a(k)\Psi_\gs\| \leq C |k|^{(\beta-3)/2 -
       |m|},\qquad |m|\leq 3.
\end{equation*}
The finiteness of \eqref{eq:bounding theta2} now follows from a
straightforward extension of Lemma A.1 in
\cite{deroeckkupiainenpropagation} to functions taking values in a
Hilbert space. This finishes the proof of Lemma \ref{lem: localization
  ground state}.

\bibliographystyle{alpha}
\bibliography{mylibrary11,ac-lit}

\end{document}